\documentclass[letterpaper]{article}
\author{Ioana Ivan\thanks{MIT Computer Science and Artificial Intelligence Laboratory, ioanai@mit.edu.}
\and
Michael Mitzenmacher\thanks{Harvard University,
School of Engineering and Applied Sciences, michaelm@eecs.harvard.edu. This work was supported by NSF grants CCF-0915922 and IIS-0964473.}
\and
Justin Thaler\thanks{Harvard University,
School of Engineering and Applied Sciences,
jthaler@seas.harvard.edu. Supported by the Department of Defense (DoD) through the National Defense Science \& Engineering Graduate Fellowship (NDSEG) Program, and in part by NSF grants CCF-0915922 and IIS-0964473.}
\and
Henry Yuen\thanks{MIT Computer Science and Artificial Intelligence Laboratory, hyuen@csail.mit.edu. Supported by an MIT Presidential Fellowship.}
}
\date{}
\title{Continuous Time Channels with Interference}

\usepackage{amsfonts}
\usepackage{fullpage}
\usepackage{times}
\usepackage{verbatim}
\usepackage{amsmath}
\usepackage{amsthm}
\usepackage{ifsym}
\usepackage{graphicx}   % need for figures
\usepackage{subfigure}  % use for side-by-side figures
\usepackage{fullpage}

\usepackage{algorithm}
\usepackage{algpseudocode} % tlmgr install algorithmicx

\newtheorem{theorem}{Theorem}[section]

\begin{document}
\maketitle
%\IEEEpeerreviewmaketitle
\begin{abstract}
Khanna and Sudan \cite{KS11} studied a natural model of continuous
time channels where signals are corrupted by the effects of both
noise and delay, and showed that, surprisingly, in some cases both
are not enough to prevent such channels from achieving
unbounded capacity. Inspired by their work, we consider 
channels that model continuous time communication with adversarial
delay errors. The sender is allowed to subdivide
time into an arbitrarily large number $M$ of micro-units in which binary
symbols may be sent, but the symbols are subject to unpredictable
delays and may interfere with each other.  We 
model interference by having symbols that land in the same
micro-unit of time be summed, and we study $k$-interference channels, which allow
receivers to distinguish sums up to the value $k$.  We consider both
a channel adversary that has a limit on the maximum number of steps it can delay 
each symbol, and a more powerful adversary that only has a bound on the
average delay.

We give precise characterizations of the threshold between finite and
infinite capacity depending on the interference behavior and on the
type of channel adversary: for max-bounded delay, the threshold is
at $D_{\text{max}}=\Theta\left(M \log\left(\min\{k,
M\}\right)\right)$, and for average bounded delay the threshold is
at $D_{\text{avg}} = \Theta\left(\sqrt{M \min\{k,
  M\}}\right)$. 
\end{abstract}
\vspace{-1mm}
\section{Introduction}
% \vspace{-1mm}
\label{sec:intro}
We study continuous time channels with adversarial
delay errors in the presence of interference.  Our models are inspired
by recent work of Khanna and Sudan \cite{KS11}, who studied
continuous-time channels in the presence of both delay errors and
(signal) noise errors.  In this model, the communicating parties can
subdivide time as finely as they wish.  In each subdivided unit of
time a 0 or 1 can be sent, but the sent signals are subject to
unpredictable \emph{delays}.  Khanna and Sudan 
found (suprisingly) that the channel capacity in their
model is finitely bounded only if at least one of the two
sources of error (delay or signal noise) is adversarial. However,
they assumed that at any instant in time, the receiver
observes the \emph{sum} of the signals delivered.

In this paper, we observe that the behavior of the channel changes
dramatically if one accounts for the possibility of interference, 
and that this holds even in the
absence of signal noise.  Our model of interference is very simple;
the symbols received at each time unit are summed, and the receiver sees the exact sum
if it is less than $k$, but values greater than $k$ cannot be distinguished from each other.  %When $k = 2$, this
%corresponds to a natural setting where the receiver can distinguish
%between no 1 signal, a clean single 1 signal, and interference from
%2 or more one signals.

At a high level, our results are two-fold. First, we show that delay
errors in the presence of interference are surprisingly powerful.
% which may explain in part why we do not observe infinite capacity in
% physical channels.  
Second, in the context of delay errors with
interference, we find that seemingly innocuous modeling decisions can
have large effects on channel behavior.

\medskip
\vspace{-1mm}
\noindent \textbf{Related Work.}
%showed in \cite{shannon1, shannon2} that in the discrete time setting, the capacity of the channel is finite
Typically a communication channel is modeled as follows. The channel takes as input a signal $f$, modeled as a function from some domain $\mathcal{T}$ to some range $\mathcal{R}$,
and the channel outputs a received signal $\tilde{f}: \mathcal{T} \rightarrow \mathcal{R}$, which is a noisy version of $f$. For discrete time channels, $\mathcal{T}$ is a finite domain $\{0, \dots, T-1\}$ where
$T$ is the time duration, and for continuous-time channels, $\mathcal{T}$ is a continuous domain such as the interval $[0, T]$. For discrete signal channels, $\mathcal{R}$ is
a finite set such as $\{0, 1\}$, and for continuous signal channels, $\mathcal{R}$ is an infinite set such as the interval $[0, 1]$.

Shannon showed in the discrete time setting, the capacity of the channel is finite even if the signal is continuous, as long as there is signal noise \cite{Shannon49}. %Specifically,
%the capacity is proportional to $\frac{1}{\sigma^2}$ where $\sigma^2$ is the variance of the noise.
Nyquist \cite{Nyquist24} and Hartley \cite{Hartley28} showed that even in the continuous time setting, the capacity is finite if one places certain restrictions on the Fourier spectrum of the signal. %if the signal can be written as a linear combination of sinusoids whose frequency
%is an integer multiple of some minimal frequency, and upper bounded by some maximal frequency.

Most relevant to us is recent work by Khanna and Sudan \cite{KS11},
which introduced continuous-time channels with signal noise and delay
errors.  They modeled their channel as the limit of a discrete process,
% (making no Nyqist-Hartley style assumption on the structure of the signal), 
and found that the capacity of their channel is infinite
unless at least one of the error sources is adversarial.

% There has also been recent work in the computer science theory
% community which is similar to ours in spirit in that they study the
% power of various computational processes in the presence of
% interference. In particular, Aspnes \emph{et al.} introduce the notion
% of $k^+$-decision trees, a generalization of standard decision trees
% where one can query multiple literals at a time, but the query answers
% cannot distinguish cases when $k$ or more of the literals are true
% \cite{k+}.

Our work differs from previous work in several ways.  We consider
channels which introduce delays adversarially, but we additionally
consider a very simple model of interference.  We also consider two
limitations on the adversary: one where maximum delay for any symbol
is bounded, and one where the average delay over all symbols is
bounded.  In both cases, we find that our channels display a clean
threshold behavior.

We believe that the adversarial setting presented here offers a clean initial analysis of the interference model, 
already with surprising results. A next natural step would be to analyze 
the effect of \emph{random} delays in the presence of interference, and we leave this question as an interesting
direction for future work.

\section{Model and Summary of Results}
% \vspace{-1mm}
\noindent \textit{Modeling Time.} Following \cite{KS11}, we model continuous time as the limit of a discrete process.
More specifically, the sender and receiver may send messages that last
a duration of $T$ units of time, but also can divide every unit of
time into $M$ subintervals called \emph{micro-intervals}, and the
sender may send one bit per micro-interval. We refer to $M$ as the
\emph{granularity of time}, and refer to a sequence of $M$
micro-intervals as a \emph{macro-interval}. We call $T$ the
\emph{message duration} of the channel.  A codeword $c$ sent over
the channel is therefore represented as $c\in\{0,1\}^{MT}$.

\medskip \noindent 
\textit{Modeling Delays.} The effect of the channel on a
sent codeword $c\in\{0,1\}^{MT}$ is to delay symbols of $c$ by some
amount, e.g. the $i$th symbol of $c$ may be moved to the $j$th
timestep of the received codeword, where $j \geq i$. The delay process
is adversarial, where we assume that the adversary knows the
encoding/decoding scheme of the sender and receiver, and both the
symbols that get delayed and the amount they are delayed can depend on
the codeword that is sent.  We formalize the notions of
\emph{max-bounded delay} and \emph{average-bounded delay} below. 

\medskip \noindent 
\textit{Modeling Interference.}  If multiple symbols are
delivered at the same time step, there are several natural ways the
channel could behave. In \cite{KS11}, the receiver observes the sum of
all bits delivered at that instant of time;  we call this the \emph{sum channel}.
%  where the bits are treated
% as integers; we call this the \emph{sum channel}. Another obvious
Another obvious
choice is for the receiver to see the {\sf OR} of all bits delivered
at that instant in time; we call this the {\sf OR} channel.

We generalize these two models to what we call the
\emph{$k$-interference} channel.  If there are fewer than $k$ 1s
delivered at an instant in time, the receiver will see the exact
number of 1s delivered; otherwise the receiver will only see that at
least $k$ 1s have arrived. Thus, the sum channel can be viewed as the
$\infty$-interference channel, and the OR channel as the
$1$-interference channel. We consider $k$-interference channels as $k$
varies between the extremes of $1$ and $\infty$, and may depend on the
granularity of time $M$. We call the parameter $k$ the \emph{collision
  resolution} of the channel.
% Similar to \cite{k+}, we consider the
% The 2-interference channel is of particular interest, as it captures the
% setting where if two or more 1-signals are received, the
% receiver can only detect that there is a collision.

% For brevity, we will denote the average-bounded delay,
% $k$-interference channel as the $\mbox{AVG-}k$ channel, and the
% max-bounded delay, $k$-interference channel as the $\mbox{MAX-}k$
% channel.

\medskip \noindent 
\textit{Valid Codebooks.} For any fixed channel and
codeword $c$, we let $B(c)$ denote the set of possible received
strings corresponding to $c$. For any time $T$, we say a codebook $C
\subseteq \{0, 1\}^{MT}$ is \emph{valid} for a channel if for any $c
\neq c'$ in $C$, $B(c) \cap B(c') = \emptyset$. Informally, this means
that the adversary cannot cause the decoder to confuse $c$ with $c'$
for any other codeword $c'$.

\medskip \noindent 
\textit{Rate and Capacity.} For any fixed granularity of time $M$ and
time $T$, let $s_{M, T} := \log |C(M, T)|$, where $|C(M, T)|$ denotes
the size of largest valid codebook $C(M, T) \subseteq \{0,1\}^{MT}$
for the channel. The \emph{capacity of the channel at granularity $M$}
is defined as $R(M) = \limsup_{T \rightarrow \infty} \{s_{M, T}/T\}.$
The \emph{capacity} of the channel is defined as $\limsup_{M
  \rightarrow \infty} R(M) = \limsup_{M \rightarrow \infty} \left [
  \limsup_{T \rightarrow \infty} \{s_{M, T}/T\} \right ]$. We stress
that the order of the limits in the definition of the channel capacity
is crucial, as we show in Section \ref{sec:extensions}.
% In particular, when calculating the rate of a channel, the
% sender and receiver are not allowed to let the granularity of time
% grow with $T$.  We show in Section \ref{sec:extensions} that
% interchanging the order of the limits in the definition substantially
% changes the behavior of the channel.
%This is (to us) a subtle point, and we demonstrate that for several channels considered in this paper, interchanging the order of the limits in the definition substantially changes the behavior of the channel. We are not convinced that the definition above (which is consistent with the definition in \cite{ks}) is the only natural one, and we elaborate on this point later.

%\medskip \noindent \textit{Multiple Symbols Arriving Simultaneously.}%

%Formally, we now have the following.%

\medskip
\noindent \textit{Encoding:} For every $T$ and $M$, the sender encodes $s_{T, M}$ bits as $MT$ bits by applying an encoding function $E_T : \{0, 1\}^{s_{T, M}} \rightarrow \{0,1\}^{MT}$. The encoded sequence is denoted $X_1, \dots, X_{MT}$.

\medskip \noindent 
\textit{Delay:} The delay is modeled by a delay function $\Delta: [MT] \rightarrow \mathbb{Z}^{\geq 0}$, where $\mathbb{Z}^{\geq 0}$ denotes the non-negative integers. The delay function has to satisfy a constraint depending on the type of delay channel we have:

\begin{itemize}
	\item \textit{Max-bounded delay}: For all $i\in[MT]$, $\Delta(i) \leq D_{\text{max}}$, where $D_{\text{max}}$ is the bound on the maximum delay.
	\item \textit{Average-bounded delay}: $\sum_i \Delta(i) \leq MT\cdot D_{\text{avg}}$, where $D_{\text{avg}}$ is the bound on the average delay.
\end{itemize}

\medskip \noindent 
\textit{Received Sequence}. The final sequence seen by the receiver given delay $\Delta$, is $Y_1, \dots, Y_{MT}\in \mathbb{Z}^{\geq 0}$, where $Y_i := \min\{k, \sum_{j \leq i \text{ s.t. } j + \Delta(j) = i} X_j\}$ and $k$ is the collision resolution parameter of the channel. We will ignore the symbols that get delayed past timestep $MT$.

For brevity, we use the shorthand $\mbox{AVG-}k$ channel
and $\mbox{MAX-}k$ channel, where the meaning is clear.

\subsection{Summary of Results}

We prove that in the case of max-bounded delay, the capacity is
finite if $D_{\text{max}}=\Omega\left(M \log\left(\min\{k,
M\}\right)\right)$, and infinite otherwise. 
% In particular, in the
% max-bounded delay case, the capacity of the sum-channel is infinite
% unless $D_{\text{max}}=\Omega\left(M \log M\right)$, and the capacity
% of the 1-interference or 2-interference channel in infinite unless
% $D_{\text{max}}=\Omega(M)$.  
In contrast, we prove that in the
case of average-bounded delay, the capacity is finite if
$D_{\text{avg}}=\Omega\left(\sqrt{M \cdot \min\{k, M\}}\right)$, and
infinite otherwise. 
% Thus, in the average bounded delay case, the
% capacity of the 1-interference or 2-interference channel is finite
% even if $D_{\text{avg}}=\Omega\left(\sqrt{M}\right)$.  See Table
% \ref{tab:comp}.

We also consider a number of variant channels and observe that
seemingly innocuous modeling choices cause the behavior to change
drastically.  In particular, we consider settings where the granularity of time is
allowed to grow with the message duration, and where adversarial signal noise can 
also be added.  For brevity, we provide a few specific interesting results. 

\section{Max-Bounded Delay Channel}
% In this section, we give a precise characterization of the infinite/finite capacity threshold of the $\mbox{MAX-}k$ channel, which is easier to obtain than for the
% $\mbox{AVG-}k$ channel. %, but the arguments we introduce here will serve as a useful starting point for the $\mbox{AVG-}k$ analysis.
We give a precise characterization of the infinite/finite capacity threshold of the $\mbox{MAX-}k$ channel.  Here and throughout, $k$ refers to the collision
resolution paramater, and $M$ to the granularity of time.
\begin{theorem} \label{thm:max}
If $D_{\text{max}}$ is the max-delay bound for the $\mbox{MAX-}k$ channel, then
then the capacity of the channel is infinite when $D_{\text{max}} = o\left(M \log\left(\min\{k, M\}\right)\right)$,
and the capacity is finite when $D_{\text{max}} =\Omega\left(M \log\left(\min\{k, M\}\right)\right)$.
\end{theorem}
% \begin{theorem} \label{thm:max}
% Let $D_{\text{max}}$ be the max-delay bound for the $\mbox{MAX-}k$ channel, where $k$ is the collision resolution parameter. Let $M$ be the granularity of time chosen by the sender and receiver.
%If $D_{\text{max}} = o\left(M \log\left(\min\{k, M\}\right)\right)$, then the capacity of the channel is infinite. If $D_{\text{max}} =\Omega\left(M \log\left(\min\{k, M\}\right)\right)$, the capacity is finite.
% \end{theorem}
\begin{proof}
\textbf{Infinite capacity regime.} Suppose  $D_{\text{max}} = c M \log\left(\min\{k, M\}\right)$ for $c=o(1)$ 
(here, $c$ denotes a function of $M$ that is subconstant in $M$).
% (that is, $c$ is a function of $M$ that is subconstant in $M$;  we simply use $c$ for convenience).

Assume for simplicity
that $1/c$ is an integer. Also assume that $c \geq \frac{1}{\log k}$ and $k \leq M$, as smaller values of $c$ and larger values of $k$ only make communication easier. %\textbf{II: I'm not sure this is obvious, as the delay also increases with k: maybe it would be better to say we only consider the case where $k \leq cM$, as the other case follows by a similar argument: in this case we'd just need to change the size of the blocks to $2cM \cdot log(cM)$}
We give a valid codebook of size $s=2^{T/2c}$, 
% from which we deduce that $R(M) = \omega(1)$, and thus the capacity is infinite.
showing $R(M) = \omega(1)$, and thus the capacity is infinite.
Given a message $x \in \{0, 1\}^{T/2c}$, the sender breaks the message $x$ into \emph{blocks} of length $\log k$.
The sender then encodes each block independently, using $2cM \log k$ bits for each block as described below. The resulting codeword has length $\frac{T}{2c \log k} \cdot 2cM\log k = TM$ as desired.

%Each macro-timestep will independently encode $\frac{M}{cM+1}$ bits of information, and we therefore specify % we
%space bits sufficiently far apart that the adversary cannot cause any collisions, while keeping bits sufficiently close together
%to ensure the amount of information transmitted grows with $M$.

A block is encoded as follows. Since each block is $\log k$ bits long, we interpret the block as an integer $y$, $1 \leq y \leq k$.
The sender encodes the block as a string of $2cM \log k$ bits, where the first $y \leq k$ bits in the string are 1s, and all remaining bits are 0s.
To decode the $j$'th block
of the sent message, the receiver simply looks at the $j$'th set of $2cM \log k$ bits
in the received string, and decodes the block to the binary representation of $y$, where $y$ is
the total count of 1s received in those $2cM \log k$ bits.

Since the maximum delay is bounded by $cM \log k$, and 1s only occur as the first $k \leq M \leq cM \log k$ locations of each sent block, any 1-bit must be delivered within its block. Furthermore, the count of 1 bits is preserved, because at most $k$ 1 bits collide within a block. Correctness
of the decoding algorithm follows.

%\textbf{HY: Should we make explicit why we need $\log \min \{k,M\}$ as opposed to just $\log k$, or is it clear from the combination of the infinite+finite capacity proof?}

\medskip
\noindent \textbf{Finite Capacity Regime.}
Suppose the delays have bounded maximum $D_{\text{max}} = cM
\log\left(\min\{k, M\}\right)$, with $c = \Omega(1)$.
We give an adversary who ensures that there at most $O(\log k)$ bits of
information are transmitted every $c\log k$ macro-timesteps. Thus, for
$c = \Omega(1)$, the rate is bounded above by $O(\frac{1}{c}) = O(1)$ for
all values of $M$, and hence the capacity is finite.

Assume first that $k \leq M$.
%; we will deal with the case of larger $k$ later. 
The adversary breaks the sent string into blocks of length $D_{\text{max}}$,
and delays every sent symbol to the end of its block. The adversary clearly never introduces a delay longer than $D_{\text{max}}$  micro-timesteps. Each received block can only take $k+1$ values: all bits of the received block will be 0, except for the last symbol which can take any integer value between $0$ and $k$. Thus, only $O(\log k)$ bits of information are transmitted every $D_{\text{max}}=cM\log k$ micro-timesteps, or $c \log k$ macro-timesteps, demonstrating finite capacity.

If $k > M$, then the adversary is the same as above, where the block size is $D_{\text{max}} = c M \log M$. Each received block can only take one of $cM \log M+1$ values, since all bits of the block are 0, except for the last symbol which may vary between $0$ and $c M \log M$. Thus, only $\log(c M \log M) = O(c\log M)$ bits of information are transmitted every $c \log M$ macro-timesteps, completing the proof.
\end{proof}

\section{Average-Bounded Delay Channel}
% We now study the behavior of the average-bounded delay,
% $k$-interference ($\mbox{AVG-}k$) channel.  We prove the following:
We now study the behavior of the $\mbox{AVG-}k$ channel.
% \begin{theorem} \label{thm:avg}
% Let $D_{\text{avg}}$ be the average-delay bound for the $\mbox{AVG-}k$ channel, where $k$ is the collision resolution parameter. Let $M$ be the granularity of time chosen by the sender and receiver.
% If $D_{\text{avg}} = o(\sqrt{M \min\{k, M\}})$, then the capacity of the channel is infinite. If $D_{\text{avg}} = \Omega(\sqrt{M \min\{k, M\}})$, the capacity is finite.
% 
%\end{theorem}
\begin{theorem} \label{thm:avg}
If $D_{\text{avg}}$ is the average-delay bound for the $\mbox{AVG-}k$ channel, then
then the capacity of the channel is infinite when $D_{\text{avg}} = o(\sqrt{M \min\{k, M\}})$,
and the capacity is finite when $D_{\text{avg}} = \Omega(\sqrt{M \min\{k, M\}})$.
\end{theorem}
\begin{proof}
\textbf{Infinite capacity regime.} Suppose $D_{\text{avg}} = c\sqrt{Mk}$, where $c=o(1)$ (that is, again, $c$ is a function of $M$ that is subconstant in $M$). Let $T$ be the message duration. Assume without loss of generality
that $c \geq \frac{1}{Mk}$ and $k \leq M$ (smaller values of $c$ and larger values of $k$ only make communication easier). 

%As in \cite[Lemma 4.1]{KS11}, we use a \emph{concatenated code}. We first encode a string $x \in \{0,1\}^{s_{T,M}}$ with a classical error-correcting %code of constant rate, which can handle up to a $1/5$-fraction of adversarial errors ($1/5$ can be replaced with any constant less than $1/4$). We %then include extra error correction on top of the classical code, which is tailored for resilience against delay errors.

Suppose the sender wants to send a message $x\in\{0,1\}^{s_{T,M}}$
with $s_{T,M} = T/c$. As in \cite[Lemma 4.1]{KS11}, we use a \emph{concatenated
code}: we assume that $x$ has already been encoded under a classical
error-correcting code $C$ that corrects a $1/5$-fraction of
adversarial errors (or any other constant less than $1/4$), as this
will only affect the rate achieved by our scheme by a constant factor.
$C$ is then concatenated with the following inner code, which is
tailored for resilience against delay errors: each bit of $x$ gets
encoded into a block of length $2\ell = 2cM$: 0's map to $2\ell$ 0's
(called a $0$-block), and 1's map to $\ell$ 1's followed by $\ell$ 0's
(called a $1$-block). The resulting codeword is thus $MT$ symbols long
as required.

For decoding, let $Y=Y_1,\ldots,Y_{MT}$ be the received word.  The receiver divides $Y$ into blocks of length $\ell$. Let $\gamma(i) = \sum_{j\in[i MT,\ldots,(i+1)MT - 1]} Y_j$ denote the number of $1$s encountered in the $i$th block. The receiver decodes $Y$ as a message $y\in\{0,1\}^{s_{T,M}}$ where $y_i$ is declared to be $1$ if $\gamma(i) \geq \sqrt{\ell k}$, $0$ otherwise. Notice $\sqrt{\ell k} \geq 1$. Finally, the receiver will decode $y$ using the outer decoder to obtain the original message. By the error-correcting properties of the outer code $C$, it suffices to show that at least $4/5$ths of the inner-code blocks get decoded correctly.

We use a potential argument to demonstrate that the adversary can afford to corrupt a vanishingly small fraction of the blocks. We maintain a potential function $\Phi(i)$ that measures the total amount of delay the adversary can apply after performing the $i$th action (where an action is delaying a single symbol some distance). Initially, $\Phi(0) = MTD_{\text{avg}}$. %We show that for any sequence of delay actions, the maximum fraction of blocks the adversary can corrupt before $\Phi$ becomes $0$ is much less than $1/5$, and hence the receiver can uniquely recover the original sent message. 

Turning a 0-block into a 1-block requires the adversary to delay at least $\sqrt{\ell k}$ $1$ symbols from some previous block at least a distance $\ell/2$, so this requires reducing $\Phi$ by $\Omega(\ell^{3/2} \sqrt{k})$.  To turn a 1-block into a 0-block, the adversary can either 1) move $1$ symbols out of the 1-block (\emph{evicting} 1s), or 2) \emph{collide} $1$s within the 1-block, or 3) a combination of both. We show that any combination requires reducing $\Phi$ by $\Omega(\ell^{3/2} \sqrt{k})$ as well.

Suppose the adversary chooses to corrupt a 1-block by evicting $\delta$ 1 symbols, and colliding the remaining $1$ symbols so that at most $\sqrt{\ell k}$ 1s remain. The adversary minimizes the amount of delays it spends to do this by evicting the last $\delta$ $1$s from a block, and choosing $\alpha$ equally spaced ``collision points'' (CPs) within the remaining $1$s, where each remaining $1$ symbol is delayed to the nearest CP ahead of it. Evicting $\delta$ 1s out of the block requires the adversary to spend at least $\delta\ell$ delays. Each CP receives $(\ell - \delta)/\alpha$ 1 symbols, and the amount of delays spent per CP is $1 + 2 + \cdots + (\ell - \delta)/\alpha = \Theta\left(\frac{(\ell - \delta)^2}{\alpha^2}\right)$. Thus, the total amount of delay spent by the adversary to corrupt the 1-block is $\Omega\left(\frac{(\ell - \delta)^2}{\alpha} + \delta \ell \right)$. This is minimized when $\delta = 0$, i.e. when no symbols are evicted. Since $\alpha k \leq \sqrt{\ell k}$ (because each CP will have value $k$ in the received string if at least $k$ 1s are delivered at that index), the adversary needs to use $\Omega(\ell^{3/2}\sqrt{k})$ units of potential in order to corrupt a 1-block.

In our analysis, the minimum potential reduction $\Omega(\ell^{3/2} \sqrt{k})$ accounts for corrupting at most a block and its adjacent neighbor. Thus, the maximum number of blocks corruptable is $2 \Phi(0)/\Omega(\ell^{3/2} \sqrt{k}) = O(c (M/\ell)^{3/2} T)=O(T/\sqrt{c})$.
%The budget of delays is $MT\cdot D_{avg} = c \sqrt{k} M^{3/2} T$. The adversary is required to
%spend $\Omega(\ell^{3/2} \sqrt{k})$ delays from its budget in order to corrupt $O(1)$ blocks, so the maximum number of blocks corruptable is $O(c %(M/\ell)^{3/2} T)$, which is $O(cT/g^{3/2})$. 
Since the original codeword had a total of $T/c$ blocks, the maximum fraction
of blocks corruptable is $O(\sqrt{c})$. However, $c = o(1)$, so a vanishingly small fraction of blocks are corrupted, 
and the original message can be recovered. Thus, we have constructed a valid codebook of size $2^{\Omega(T/c)}$, and this implies that the capacity is infinite. 

\medskip \noindent \textbf{Finite capacity regime.} Suppose the delays have bounded average $D_{\text{avg}} = c\sqrt{M \min\{k, M\}}$, for some constant $c$.
We will assume for simplicity that $c=1$ and $k \leq M$, and explain how to handle smaller values of $c$ and larger values of $k$ later.
We show the capacity is finite by
specifying an adversary who ensures that there are a constant number of possible received strings
for almost every macro-timestep.

To accomplish this, the adversary will break the sent string into blocks of length $M$. It scans the blocks sequentially, and adds and removes 1s so that each block will have 1s only at indices that are multiples of $D_{\text{avg}}$, or at the very last index of the block. The adversary ensures that it can always add 1s when it needs to by
maintaining a ``bank'' of delayed 1s from previous blocks that will have size between $D_{\text{avg}}$ and $2D_\text{avg}$ 1s whenever possible.
The bank will always be small enough so that it does not contribute too many delays to the average. 
Once the bank reaches size $D_{\text{avg}}$, its size never falls below this level again.
%guaranteeing that the corrected adversary never runs into the problems encountered above.
%Until the bank reaches this size, the communicating parties may be able to communicate a fair amount of information,
We show that the amount of information transmitted before the bank reaches
this size is negligible for large $T$.

%is possible that will run into the same problem as the simplified adversary,
%in which case the macro-timestep may contain a fair
%amount of information, but this can only happen once. Thus, for
%any fixed $M$, as the number of macro-timesteps $T$ approaches infinity, the rate
%will go to 0.

% It will ensure that (except for
%possibly the first block containing more than 2M^{1/2} 1s) each
%received block either is all-0s except for a 1 in the last location,
%or has a 1 at every location which is an integer multiple of M^{1/2}.
%Thus, any block other than maybe the first conveys at most 1 bit of
%information.

The adversary considers each block in turn, and its actions falls into four cases. Let $\ell$ denote the number of 1s in the block, and let $s$ 
denote the size of the bank at the start of the block.

\begin{enumerate}
\item If $\ell \leq D_{\text{avg}}$ (we call the block \emph{light}):
	\begin{enumerate}
	\item If $s \geq D_{\text{avg}} + k - \ell$, the adversary will delay all 1s in the block until the final index within the block. 
	If $\ell < k$, the 
	adversary will also deliver $k-\ell$ 1s from the bank at the final index to ensure that the value of the final index is $k$. When this step completes, 
	the size of the bank will be between $D_{\text{avg}}$ and $s$ .
	\item If $s < D_{\text{avg}} + k - \ell$, the adversary adds all 1s in the block to the bank, ensuring that the received block consists entirely of 0s.
	When this step completes, the bank has size least $s$ and at most $D_{\text{avg}} + k \leq 2 D_{\text{avg}}$.
% , where the inequality holds because $D_{\text{avg}}$ is the geometric
%	average of $k$ and $M$, and we have assumed $k \leq M$.
	\end{enumerate}
	
\item If $\ell > D_{\text{avg}}$ (we call the block \emph{heavy}):
	\begin{enumerate}
	\item  If $s\leq D_{\text{avg}}$, the adversary adds $D_{\text{avg}}-s < \ell$ of the new 1s to the bank, and it delays the rest of the 1s
to the nearest integer multiple of $D_{\text{avg}}$. 
	\item Otherwise, $s$ will be at least $D_{\text{avg}}$. The adversary will place $k$ 1s at every location which is an integer
multiple of $D_{\text{avg}}$ using bits from its bank (this requires at most $k M/D_{\text{avg}} = k M/\sqrt{kM} = D_{\text{avg}}$ bits), and delays the first
$\ell - D_{\text{avg}}$ 1s within the current block to the nearest integer multiple of $D_{\text{avg}}$.
The last $D_{\text{avg}}$ 1s get added to the bank to replace the 1s lost from
the bank, so the bank stays at size $s$.
	\end{enumerate}
\end{enumerate}

We argue that at most $\sqrt{Mk} \log k+ O(T)$ bits of information are
transmitted over $T$ blocks by the above scheme.  Once the bank
reaches size $D_{\text{avg}}$, there are only three possible values
for each received block: the all-zeros vector; the vector that is all
0s except for the final index which has value $k$; and the vector that
is all 0s except for indices which are integer multiples of
$D_{\text{avg}}$, which have value exactly $k$.  Before the bank
reaches size $D_{\text{avg}}$, any light block is still received as
either the first or second possibility just described.  
% All that
%remains is to observe that at most one heavy block is encountered
% before the bank reaches size $D_{\text{avg}}$, and there are at most
% $k^{M/D_{\text{avg}}} \leq k^{\sqrt{M\cdot k}}$ possible values this
% block can take, since only indices which are integer multiples of
% $D_{\text{avg}}$ are non-zero in the received string. 
Finally, at most one heavy block is encountered before the bank reaches size $D_{\text{avg}}$, and this block
can take on at most $k^{M/D_{\text{avg}}} \leq k^{\sqrt{M k}}$ possible values. 
Thus, over all
$T$ blocks, at most $\sqrt{Mk} \log k+ O(T)$ bits of information are
transmitted,
%. As $T \rightarrow \infty$, the rate is clearly bounded,
and hence the capacity is finite.
%{\bf MM:  Maybe I'm being dense, but if O(T) bits are transmitted, and we just get the capacity by dividing by $T$, how is
%this going to 0?  Do you just mean it's bounded?}

%responsible for at most $\frac{1}{2} \log M$ bits of information and the rest are
%responsible for at most 1 bit of information. 

Finally, we bound the average delay incurred by the adversary. For
each block, we separately bound the total delays incurred by the
symbols banked at the beginning of the block and symbols within the block. The
symbols within any light block are responsible for total delay at most
$MD_{\text{avg}} $, since at most $D_{\text{avg}}$
symbols are delayed at most $M$. 
% Similarly, the 
The symbols within in any heavy block are responsible for total delay at most $2M
D_{\text{avg}}$, since all but at most $D_{\text{avg}}$ 1s are
delayed only until the nearest integer multiple of $D_{\text{avg}}$,
and the rest are delayed at most $M$. As the bank contains at most
$2 D_{\text{avg}}$ 1s, banked symbols contribute at most $2
MD_{\text{avg}}$ total delays per block. The
adversary therfore spends at most $4 M D_{\text{avg}}$ total delays per
block, for an average delay of $4 D_{\text{avg}}$. To reduce this to  $D_{\text{avg}}$, 
we modify the above construction to use a block
length of $M/16$ micro-timesteps, decreasing the average delay appropriately
while increasing the rate by only a constant factor.

It remains to explain how to handle cases $c < 1$ and $k > M$. If $c <
1$, we simply decrease the block size further, from $M/16$ to
$Mc^2/16$. This decreases the average delay 
by a factor of $c$ and increases the rate by only a
constant factor.  For $k > M$, we note the adversary described above
never delivers more than $M$ 1s at any particular
micro-timestep. Thus, even if $k > M$, the 
the received string is the same as it would be if $k=M$.
\end{proof}

\section{Extensions and Alternative Models}
\label{sec:extensions}
% In this section, we consider some natural variants of the
% $\mbox{AVG-}k$ and $\mbox{MAX-}k$ channels, and observe that seemingly
% innocuous modeling choices result in drastically different channel
% behaviors.  Our first result elucidates a subtle point alluded to in
% Section \ref{sec:intro}: namely, the order of the limits in the
% definition of channel capacity matters. We show that, if the
% granularity of time $M$ is allowed to grow with the message duration
% $T$, then delays become significantly more benign.
% 
% Our second result considers the interaction of \emph{noise} with delays in the presence of interference.
% We give a simple construction demonstrating that the combination of interference with noise yields a max-bounded adversary that is surprisingly potent.

\subsection{The Order of the Limits Matters}
%{\bf MM:  I find this paragraph hard to follow and think it should be changed, but I don't know how.  The single macro-step can dominate
%is what seems confusing.  Is there a way to explain this better?}
Under the definition of capacity used in the sections above and in \cite{KS11},
$\limsup_{M \rightarrow \infty} \limsup_{T \rightarrow \infty}  \{k_{M, T}/T\}$,
the sender and receiver are not allowed to let the granularity of time grow with $T$. If we instead define the capacity to be
$\limsup_{T \rightarrow \infty} \limsup_{M \rightarrow \infty}  \{k_{M, T}/T\}$,
then the channel would behave very differently. 
Conceptually, the reason is that if $M$ is allowed to grow with $T$, the sender and
receiver can choose $M$ to be so much larger than $T$ that
a vast amount of information (relative to $T$) can be encoded in just the first macro-timestep, avoiding interference issues. 
% This prevents the adversary from using 1s from earlier macro-timesteps to cause interference.
%The reason this is so powerful is that it prevents the adversary from having 1s from early macro-timesteps interfere with later timesteps. In particular,
% Notice in particular that this renders useless the adversary described in the finite capacity regime of Theorem \ref{thm:avg}, as this adversary allowed 
% a significant amount of information to be transmitted in the first \emph{heavy} macro-timestep. 

To demonstrate one place where this interchange of limits alters the channel capacity, we show the $\mbox{AVG-}1$ channel behaves differently under this definition. 

% To demonstrate that this interchange of limits alters the behavior of the channel, we show that under this new definition of capacity, the
% $\mbox{AVG-}1$ channel behaves similarly to the $\mbox{MAX-}1$ channel. That is,
% the adversary must be allowed to induce very large delays on average in order to ensure finite capacity.

\begin{theorem} If one interchanges the order of limits in the definition of channel capacity, then the capacity of
the $\mbox{AVG-1}$ channel with %average delay 
$D_{\text{avg}}=o(M)$ is infinite. \end{theorem}
\begin{proof} The idea is that the sender encodes $\omega(1)$ bits of information via the location of the \emph{first 1} in the entire codeword.
More formally,
suppose $D_{\text{avg}}=cM-1$ with $c=o(1)$, and let $c'=\sqrt{c/2}$. Assume for simplicity that $Mc'$ is an integer.
We will construct a valid codebook $\mathcal{C} \subseteq \{0, 1\}^{MT}$  with $|\mathcal{C}| = \Omega(1/c') = \omega(1)$
such that for each message $x \in \mathcal{C}$, the last $T-1$ macro-timesteps consist
only of 0s. Thus, we only specify the first macro-timestep in each codeword $x$.
In the first codeword, the first macro-timestep will simply be $Mc'$ 0s followed by $M-Mc'$ 1s.
In the second codeword, the first macro-timestep will be $2Mc'$ 0s followed by $M-2Mc'$ 1s.
In general, in the $i$th codeword, the first macro-timestep will be $iMc'$ 0s followed by $M-iMc'$ 1s.

The decoder will look at the position $L$ of the left-most 1 in the
received string and output the largest $i$ such that $iMc' \leq L$.

In order for the adversary to force the decoder
to decode incorrectly, the decoder has to make the first 1 appear at least $c' \cdot M$
positions later than it does in the sent string. For this to happen,
the adversary has to spend at least $1+ 2+ \cdots + Mc' \geq M^2c'^2/2$ delays in total.
So the average delay has to be at least $\frac{Mc'^2}{2T}=\frac{Mc}{4T}$.
For fixed $T$, this is $\Omega(M)$.  
\end{proof}

\subsection{Adding noise}

In this section we note that the \emph{combination} of interference with noise yields a max-bounded adversary that is surprisingly potent.
%We omit the proof of the statement for lack of space (it appears in the arXiv version).  
% Interestingly, the introduction of noise does not substantially alter the finite capacity threshold in the case of average-bounded delay, or in the case of max-bounded delay with no interference (i.e. the \text{SUM}-channel).
% We omit the analysis for these latter cases for brevity.

\begin{theorem} Suppose the adversary is allowed to flip $t$ bits per macro-timestep, and delay each bit a maximum of $D_{\text{max}}$ micro-timesteps. Then the capacity of the 1-interference channel
is finite if $D_{\text{max}} \cdot t = \Omega(M)$, and is infinite if $D_{\text{max}} \cdot t = o(M)$. In particular, the capacity is finite if $t=D_{\text{max}}=\Omega(\sqrt{M})$. \end{theorem}
\begin{proof}
\textbf{Finite capacity regime.} Suppose for simplicity that $D_{\text{max}} \cdot t = M$. The key observation is that the adversary can turn \emph{any} macro-timestep into
the unique string consisting of all zeros, except for 1s at indices which are integer multiples of $D_{\text{max}}$, while staying within its budget.
The adversary breaks each macro-timestep into $t$ blocks of length $D_{\text{max}}$, and delays each
bit to the end of its block. If no 1s are sent in a
block, the adversary flips a single bit in the block to create a 1.
This totals at most $t$ bit-flips per macro-timestep, giving the result.

%{\bf MM:  I assume the $d$'s everywhere are supposed to be $D_{\text{max}}$?}

\medskip
\noindent \textbf{Infinite capacity regime.} Suppose $D_{\text{max}} \cdot t = o(M)$. We again use a concatenated code to construct
a valid codebook of size $2^{\Omega\left(\frac{MT}{D_{\text{max}}}\right)}$.
The sender starts with a string $x \in \{0, 1\}^{\frac{MT}{D_{\text{max}}+1}}$ encoded under a classical error-correcting code with constant rate which can tolerate up to a $1/5$ fraction
of adversarial errors.

The sender then replaces each bit of $x$ with a block of length
$D_{\text{max}}+1$: if $x_i=0$, the $i$'th block is set to the
all-zeros string, and if $x_i=1$, the $i$'th block is set to the
string consisting of a 1 followed by $D_{\text{max}}$ zeros. The
decoder decodes a block to 1 if the block contains at least one 1, and
decodes a block to 0 otherwise. It suffices to show that at
least a $4/5$ fraction of the blocks get decoded correctly.

Call a block \emph{dirty} if even a single bit within it is flipped, and call the block \emph{clean} otherwise.
Since the adversary can afford to flip only $T \cdot t = o(\frac{MT}{D_{\text{max}}})$ bits of the course of the entire message, only an $o(1)$ fraction of blocks are dirty.
For any clean block, the adversary cannot afford to delay the first bit in the block beyond the final bit in the block. Thus, any clean block will get
decoded correctly. It follows that decoding will always be successful, yielding the result.
\end{proof}

\section{Discussion}
We studied a variety of natural models for continuous-time channels
with delay errors in the presence of interference. Our results show
that these channels exhibit a clean threshold behavior.
% , and that the
% capacity is finite even for relatively small average delay. 
% Moreover,
We note our finite capacity results hold even for \emph{computationally simple} adversaries in the
sense of Guruswami and Smith \cite{GS10};  that is, our adversaries
process the sent string sequentially in linear time, using just $O(M)$ space.
Our results can be viewed as a counterweight to those of Khanna and Sudan \cite{KS11}, by showing that
other natural additional restrictions can lead to finite capacity in their
model.

Many questions remain for future work. %Although our finite capacity
%%results rely only on computationally simple adversaries, requiring
%reliable communication in the presence of adversarial behavior is
%still rather extreme.  
Our results only address adversarial delays;  random delays
under our interference model remains open.  
% and determining 
% the capacity of (variants of) the
% $k$-interference channel in the presence of random delays
% remains open. 
One might also consider different models of
interference, or different limitations on the delays introduced by the
adversary.

\bibliographystyle{abbrv}
\bibliography{biblio}

\end{document}